\documentclass[submission,copyright,creativecommons]{eptcs}
\providecommand{\event}{TARK 2023} 
\usepackage{savesym}
\savesymbol{amalg}

\usepackage{amsmath}
\usepackage{amsthm}
\usepackage{amsfonts}
\usepackage{amssymb}
\usepackage{mathabx}
\restoresymbol{TXF}{almag}
\usepackage{mathptmx}

\usepackage{iftex}

\ifpdf
  \usepackage{underscore}         
  \usepackage[T1]{fontenc}        
\else
  \usepackage{breakurl}           
\fi

\usepackage{enumitem}
\usepackage[pagewise]{lineno}

\usepackage{enumitem}
\setlist[enumerate,1]{label=(\arabic*), ref=(\arabic*)}
\setlist[enumerate,2]{label=\alph*), ref=(\arabic{enumi}\alph*)}

\newtheorem{theorem}{Theorem}
\newtheorem{fact}{Fact}
\newtheorem{lemma}{Lemma}

\newtheorem{definition}{Definition}

\title{An Acceptance Semantics for Stable Modal Knowledge\\
\normalsize Extended Abstract}
\author{Peter Hawke
\institute{Philosophy Department, Lingnan University, Hong Kong}
\email{peterhawke@ln.edu.hk}
}
\def\titlerunning{Acceptance Semantics for Stable Modal Knowledge}
\def\authorrunning{Peter Hawke}
\begin{document}
\maketitle

\begin{abstract}
We observe some puzzling linguistic data concerning ordinary knowledge ascriptions that embed an epistemic (im)possibility claim. We conclude that it is untenable to jointly endorse both classical logic and a pair of intuitively attractive theses: the thesis that knowledge ascriptions are always veridical and a `negative transparency' thesis that reduces knowledge of a simple negated `might' claim to an epistemic claim without modal content. We motivate a strategy for answering the trade-off: preserve veridicality and (generalized) negative transparency, while abandoning the general validity of contraposition. We survey and criticize various approaches for incorporating veridicality into \textit{domain semantics}, a paradigmatic `information-sensitive' framework for capturing negative transparency and, more generally, the non-classical behavior of sentences with epistemic modals. We then present a novel information-sensitive semantics that successfully executes our favored strategy: \textit{stable acceptance semantics}.
\end{abstract}

\section{\textit{Introduction}}\label{puzzle-strategy}

In this paper, we are concerned with the semantics and logic of ordinary knowledge ascriptions that embed an epistemic (im)possibility claim.

\begin{enumerate}
\item Ann knows that it might be raining.\label{opening-example2}
\item Ann knows that it can't be raining.\label{opening-example3}
\end{enumerate}

It is natural to interpret the modals here as having an \textit{epistemic} flavor. Intuitively, \ref{opening-example2} communicates (perhaps \textit{inter alia}) that Ann's knowledge leaves it open that it is raining; \ref{opening-example3} communicates (perhaps \textit{inter alia}) that Ann's knowledge rules out that it is raining. In support, notice how jarring the following sound:

\begin{enumerate}[resume]
\item \# Ann knows that it might be raining and Ann knows that it isn't raining.\label{opening-example4}
\item \# Ann knows that it can't be raining and for all Ann knows, it is raining.\label{opening-example5}
\end{enumerate}

Note that \ref{opening-example2} and \ref{opening-example3} also provide evidence of the systematic \textit{shiftiness} of ordinary epistemic modals. Compare a bare might claim:

\begin{enumerate}[resume]
\item It might be raining.\label{opening-example2a}
\end{enumerate}

In this case, the modal is most naturally taken to communicate that the knowledge of the \textit{speaker} (who need not be Ann) leaves it open that that is raining. As evidence, note the incoherence of the following so-called (and much discussed) \textit{epistemic contradiction} (cf. \cite{Veltman1996},\cite{Yalcin2007}).

\begin{enumerate}[resume]
\item \# It might be raining and it isn't raining.\label{opening-example4a}
\end{enumerate}

The first aim of the present paper is to highlight some unusual and subtle logical features that attitude ascriptions like \ref{opening-example2} and \ref{opening-example3} plausibly display (\S \ref{linguistic-data} and \S \ref{section-strategy}), in particular in interaction with bare modal claims like \ref{opening-example2a}. The second aim is to propose a novel formal semantics that successfully predicts these features (\S \ref{section-bluntness-redux}), in contrast to a salient rival theory (\S \ref{section-rivals}). The resulting theory is of linguistic, technical, and philosophical interest. On the linguistic side, we combine novel and known linguistic data to motivate a new entry in the tradition of `information-sensitive' semantics for ordinary epistemic modals (cf. \cite{Veltman1985}, \cite{Veltman1996}, \cite{Yalcin2007}, \cite{MacFarlane2014}, \cite{Hawke2016}, \cite{Hawke2021}, \cite{Aloni2022}), extending a standard `state-based' account with a novel semantics for knowledge ascriptions. On the technical side, our system displays intriguing and striking non-classical logical behavior, motivating a fuller technical study of the underlying epistemic logic and its interactions with modals (cf. \cite{Dabrowski1996}, \cite{Puncochar2015},\cite{Yang2017a}). On the philosophical side, our semantics may be viewed as a new development in the expressivist tradition for epistemic vocabulary (cf. \cite{Yalcin2011}) that treats assertion conditions as primary in semantics (cf. \cite{Schroeder2008a}).

\section{\textit{Linguistic Evidence for Transparency and Veridicality}}\label{linguistic-data}

We work with formal language $\mathcal{L}$, intended to formalize the relevant fragment of declarative English. We use $\varphi$ and $\psi$ for arbitrary formulas. Intuitively, read $K_a \varphi$ as `Agent $a$ knows that $\varphi$' (with $a \in \{1, 2, \ldots, n\}$) and read $\diamond \varphi$ as `It might be that $\varphi$'. We take atoms $p$ and $q$ to be declaratives without logical vocabulary (we include $\diamond$ in the logical vocabulary). We use $\vdash$ to denote entailment and $\equiv$ for logical equivalence, relative to our intended reading of $\mathcal{L}$. With this in mind, there are reasons to the think that the following principles are sound, and should be recovered by a formal semantics that aims to honor our intended reading of $\mathcal{L}$.

\medskip

\begin{tabular}{l l}
\textbf{Negative Transparency (NTrans):} & $K_a\neg \diamond p \equiv K_a\neg p$\\
\textbf{K-veridicality (Ver):} & $K_a \varphi \vdash \varphi$\\
\end{tabular}

\medskip

As evidence, note that the following bare assertions (easily multiplied) have an air of incoherence. 

\begin{enumerate}[resume]
\item \# Ann knows that Bob can't be here but, for all she knows, he is. (cf. \ref{opening-example4}) \label{example1}
\item \# Ann knows that Bob isn't here but, for all she knows, he might be. \label{example2}
\item \# Bob can't be here, but Ann mistakenly knows that he might be. \label{example3}
\end{enumerate} 

Compare \ref{example3} to the benign `Bob can't be here, but Ann mistakenly believes that he might be'. \textbf{NTrans} predicts that \ref{example1} and \ref{example2} are contradictory; \textbf{Ver} predicts that \ref{example3} is contradictory.

As further evidence, note the difficulty in distinguishing the information communicated by the following in conversation:

\begin{enumerate}[resume]
\item \# For all Ann knows, Bob is here. \label{example4}
\item \# For all Ann knows, Bob might be here. \label{example5}
\end{enumerate} 

\ref{example4} and \ref{example5} seem to say the same thing: nothing that Ann knows rules out that Bob is here. Assuming that `for all Ann knows, $\varphi$' is formalizable as `$\neg K \neg \varphi$',  \textbf{NTrans} predicts this equivalence, as it entails (with minimal further assumptions) that $\neg K \neg p$ is equivalent to $\neg K \neg \diamond p$.

Observations of the above sort are not without precedent. \textbf{Ver} is orthodox (though it is notable, as \ref{example3} seems to demonstrate, that \textbf{Ver} is undisturbed by modal content). \textbf{NTrans} is related to \L ukasiewicz' principle (i.e., $\neg p \vdash \neg \diamond p$), which is in turn related to the much-discussed incoherence of `epistemic contradictions' (i.e, claims of the form $\neg p \wedge \diamond p$ or $p \wedge \diamond \neg p$)  \cite{BledinLando2018,Yalcin2007}. 

Combining \textbf{NTrans} and \textbf{Ver} with classical logic has untoward effects. To see this, first note a seemingly benign consequence of \textbf{NTrans} and \textbf{Ver}.

\begin{fact}
\textbf{NTrans}+\textbf{Ver} entails \textbf{Epistemic \L ukasiewicz (ELuk)}: $K_a \neg p \vdash \neg \diamond p$
\end{fact}

\begin{proof} $ K_a \neg  p  \; \underset{\mathtt{NTrans}}{\vdash} \; K_a \neg \diamond p \; \underset{\mathtt{Ver}}{\vdash} \; \neg \diamond p$\end{proof}

There is \textit{prima facie} evidence that \textbf{ELuk}  is an apt principle on our intended reading of $\mathcal{L}$. Consider:

\begin{enumerate}[resume]
\item Ann knows that it isn't raining. So, it can't be raining. \label{example6}
\item Ann has conclusively established that it isn't raining. So, it must not be raining.\label{example7}
\item \# Bob knows that it isn't snowing, but it might be. \label{example8}
\item \# Bob has conclusively established that it isn't snowing, but it might be. \label{example9}
\end{enumerate} 

\ref{example6} seems like unobjectionable ordinary reasoning (to bolster this, the effect seems heightened when considering the closely related reasoning in \ref{example7}). \ref{example8} has an air of incoherence (as does the closely related \ref{example9}). \textbf{ELuk} explains both. But combining \textbf{ELuk} with unfettered classical logic has puzzling results. Consider:

\medskip

\begin{tabular}{l l}
\textbf{Double Negation (DN):} & $\neg \neg \varphi \equiv \varphi$\\
\textbf{Contraposition (Con):} & $\varphi \vdash \psi$ implies $\neg \psi \vdash \neg \varphi$\\
\end{tabular}

\medskip

\begin{fact}
\textbf{ELuk}+\textbf{Con}+\textbf{DN} entails \textbf{Uniformity I}: $\diamond p \vdash \neg K_a \neg p$
\end{fact}

\begin{proof} $ \diamond p  \; \underset{\mathtt{DN}}{\vdash} \; \neg \neg \diamond p \; \underset{\mathtt{ELuk+Con}}{\vdash} \; \neg K_a \neg p $. \end{proof}

\begin{fact}
\textbf{Uniformity I}+\textbf{Ver} entails \textbf{Uniformity II}: $K_a \diamond p \vdash \neg K_b \neg p$
\end{fact}

\begin{proof} $ K_a\diamond  p  \; \underset{\mathtt{Ver}}{\vdash} \; \diamond p \; \underset{\mathtt{Uni}}{\vdash} \; \neg K_b \neg p $.\end{proof}

\textbf{Uniformity I} and \textbf{II} seem invalid, egregiously implying that if an agent is aware of but rightly uncertain about $p$, \textit{every} agent is uncertain about $p$. To see this, note that \textbf{Uniformity I} (with minimal assumptions) entails: $\diamond p \wedge \diamond \neg p \vdash \neg K_a \neg p \wedge \neg K_a p$. But `it might be raining and might not be raining' predominantly serves to express the \textit{speaker's} ignorance about the rain, while `Jones doesn't know that it is raining and doesn't know that it isn't raining' expresses that \textit{Jones} is ignorant: it is generally agreed that $\diamond p$ either has a solipsistic reading as its default, or something close (e.g.., expression of the information state of a select group of agents that includes the speaker). Similarly, note that \textbf{Uniformity II} (with minimal assumptions) entails: $K_1 \diamond p \wedge K_1 \diamond \neg p \vdash \neg K_2 \neg p \wedge \neg K_2 p$. But `Smith knows it might be raining and might not be raining' predominantly serves to express \textit{Smith's} ignorance about the rain, while `Jones doesn't know that it is raining and doesn't know that it isn't raining' predominantly serves to express that \textit{Jones} is ignorant.

To bolster this assessment, consider a banal context. Suppose that your dinner partner has a severe allergy to shellfish. You ask your waiter, Smith, `Does the daily soup contain shellfish?'. Smith replies:

\begin{enumerate}[resume]
\item It might. The kitchen usually puts shellfish in the soup, but not always. I'll check with Chef Jones. She always knows exactly what's in the soup.\label{example10}
\end{enumerate} 

Upon hearing \ref{example10}, and waiting for Smith to return, one would normally happily accept/say all of: 

\begin{enumerate}[resume]
\item The soup might have shellfish (that's why Smith is checking with the kitchen).\label{example10a}
\item Smith knows that the soup might have shellfish.\label{example10b}
\item Unlike Smith, Jones knows whether the soup has shellfish.\label{example10c}
\end{enumerate} 

It would be odd to conclude from \ref{example10a} and \ref{example10b}, per \textbf{Uniformity}, that Jones doesn't know that the soup doesn't have shellfish. For then an uncontentious application of disjunctive syllogism, using \ref{example10c}, would yield (even before Smith returns): chef Jones knows that the soup has shellfish. Surely one shouldn't conclude \textit{this} given only \ref{example10}. 

The general pattern here is emulated by other epistemic vocabulary. Let's use $\triangledown \varphi$ for `it is likely that $\varphi$'. Then $K_a\neg p \vdash \neg \triangledown p$ (and $K_a\neg p \vdash K_a \neg \triangledown p$) is similarly well-supported by \textit{prima facie} linguistic evidence, while the contrapositive $\triangledown p\vdash \neg K_a \neg p$ does \textit{not} seem true. Compare:

\begin{enumerate}[resume]
\item Ann knows that it isn't raining. So, it isn't likely to be raining. \label{example11}
\item Ann knows that it isn't raining. So, Ann knows that it isn't likely to be raining. \label{example12}
\item It is likely to rain tomorrow, but only our local metereologist Jones knows for sure.\label{example13}
\end{enumerate} 

\ref{example11} and \ref{example12} strike me as good, if redundant, reasoning (easily generalized), while \ref{example13} seems perfectly intelligible.

One style of response to all this tries to exploit the context-sensitivity of epistemic `might' to preserve restricted versions of \textbf{Ver} and \textbf{NTrans} without abandoning classical logic. In particular, the strategy would be to say that \textbf{NTrans}, \textbf{ELuk}, and \textbf{Uniformity I} hold only when the `$\diamond$' deployed in $K_a \neg \diamond p$ and $\neg \diamond p$ is indexed to the information available to agent $a$ (i.e., the same agent referred to in $K_a \neg p$). This is best expressed by enriching the syntax for $\mathcal{L}$, to record the agent each instance of $\diamond$ is indexed to:

\medskip

\begin{tabular}{l l}
\textbf{Restricted NTrans:} & $K_a\neg \diamond_a p \equiv K_a\neg p$\\
\textbf{Restricted ELuk:} & $K_a \neg p \vdash \neg \diamond_a p$\\
\end{tabular}

\medskip

 It may then be claimed that any ill results (e.g. unrestricted \textbf{Uniformity}) leading from  \textbf{Con} and \textbf{DN} are a mere illusion brought on by subtle shifts in context. This strategy should not be dismissed out of hand. Nevertheless, its execution will not be trivial. Among other complications, it sits uneasily with the data collected above (for example, our intuitive assessment of claims \ref{example6}-\ref{example9}, in support of \textbf{ELuk}, does \textit{not} seem to hinge on taking `might'/`can't'/`must' to be indexed to Ann/Bob's information specifically) and risks introducing such loose criteria for contextual shifts that the relevant explanations become bereft of content.

To bolster the alternative strategy of dropping classical logic (at least when epistemic modals are in play), note that independent motivation for rejecting \textbf{Con} has been tabled. For example, one might think that the empirical case for \L ukasiewicz' principle is compelling (cf. \cite{Bledin2014}) and argue on this basis that \textbf{Con} must be false (given that $\diamond p \vdash p$ is obviously false). Alternatively, a proposed counterexample to modus tollens from \cite{Yalcin2012c}, utilizing `likely', is easily modified to bear against \textbf{Con}. Suppose an urn contains 100 marbles, big and small. Of the big, 10 are blue and 30 are red. Of the small, 50 are blue and 10 are red. A marble, $m$, is randomly selected and placed under a cup. Given only this information, \ref{marble1} sounds like good reasoning, but \ref{marble2} does not:

\begin{enumerate}[resume]
\item Suppose that $m$ is big. It follows that $m$ is likely to be red. \label{marble1}
\item $m$ isn't likely to be red. \# Thus, $m$ isn't big.\label{marble2}
\end{enumerate}

To see why the second inference in \ref{marble2} seems incorrect, note that we \textit{already} know that the marble isn't likely to be red, yet accepting that it isn't big is rash.

The current paper thus pursues the strategy of giving an independently motivated formal semantics that delivers \textbf{Ver} and \textbf{NTrans}, while invalidating \textbf{Uniformity (I)} and invalidating \textbf{Con}.

We add one last wrinkle to our list of logical desiderata: it seems that \textbf{NTrans} can be generalized (in ways that bear on our discussion). Consider:

\medskip

\begin{tabular}{l l}
\textbf{Generalized Negative Transparency (GeNT):} & $K_a \neg (p \wedge \diamond q) \equiv K_a \neg (p \wedge q) $\\
& $K_a (p \vee \neg \diamond q) \equiv K_a (p \vee \neg q) $\
\end{tabular}

\medskip

In both cases, \textbf{NTrans} is a special case (respectively, $p = \top$ and $p = \bot$). For convenience, I assume the above claims are equivalent (they could be deployed individually in our coming argumentation, however). Note that the linguistic evidence in support of \textbf{GeNT} seems no worse than that for \textbf{NTrans} (though, unsurprisingly, parsing the relevant sentences requires slightly more effort). Consider:

\begin{enumerate}[resume]
\item \# Ann knows that it isn't both raining and a good day for a picnic, but for all she knows it's both raining and might be a good day for a picnic. \label{example14}
\item \# Ann knows that either it isn't raining or must not be a good day for a picnic, but for all she knows it's both raining and a good day for a picnic. \label{example15}
\item Ann knows that it isn't both raining and a good day for a picnic. So, Ann knows that either it isn't raining or it must not be a good day for a picnic. \label{example16}
\end{enumerate} 

\ref{example14} and \ref{example15} sound incoherent; \ref{example16} sounds like good reasoning. \textbf{GeNT} explains all this.

\section{\textit{Strategy}}\label{section-strategy}

Altogether, our target in the current paper is this:

\begin{center}
\textit{Goal:} Provide an independently motivated formal semantics that validates \textbf{Ver} and \textbf{GeNT} (with \textbf{NTrans} as a special case), and invalidates \textbf{Uniformity I}.
\end{center}

We proceed as follows. In \S \ref{section-rivals}, we consider the \textit{domain semantics} of \cite{Yalcin2011} and \cite{MacFarlane2014}, a standard `information-sensitive' semantics for `might' claims (designed to account, in particular, for non-classical behavior induced by epistemic contradictions). Equipping domain semantics with an account of attitude ascriptions presented by \cite{Yalcin2011} (following \cite{Hintikka1962} and \cite{Dabrowski1996}) delivers \textbf{NTrans}. A natural starting point is thus to ask if \textbf{Ver} and \textbf{GeNT} can be realized in this setting without fuss. However, \textit{ad hoc} maneuvers aside, this system forces a choice between \textbf{NTrans} and \textbf{Ver}. What's more, even \textit{with} said \textit{ad hoc} maneuvers, the system fails to deliver \textbf{GeNT}. 

\S \ref{section-bluntness-redux} thus proposes a novel alternative theory, showcasing a related but distinct tradition of information-sensitive semantics: we propose a formal \textit{acceptance semantics} (in the ballpark of \cite{Veltman1985},\cite{Schroeder2008a}, \cite{Hawke2016},\cite{Hawke2021}, \cite{Ciardelli2021}, \cite{Aloni2022}) that delivers \textbf{Ver} and \textbf{GeNT} as desired. Our treatment of $\diamond p$ is essentially standard for such a framework; the more novel aspect is our account of $K\varphi$, and its interaction with $\diamond p$. The guiding idea is that knowledge ascription reflects the \textit{stability} of knowledge under \textit{available refinements} of veridical information. A notion of inter-subjective `available information' sets the bound on available refinements. A variation of a classic example (cf. \cite[pg. 148]{Hacking1967}) provides initial motivation (cf. the Schmolmes case in \cite[sect.1]{Fintel2011}):

\begin{quote}
\textbf{Salvaging Operation.} Imagine a salvage crew searching for a ship that sank a long time ago. The mate of the salvage ship works from an old log, but overlooks some pertinent entries in the log, and concludes that the wreck may be in a certain bay. He confidently says `the hulk might be in these waters'. But, as it turns out later, careful examination of the log shows that the boat must have gone down at least thirty miles further south.
\end{quote}

One hesitates to say `the mate \textit{knew} that the ship might be in the bay' (better to say `he merely believed it might be'), given that his rational acceptance of `it might be in the bay' did not survive the incorporation of readily available information.

Our semantics may thus be taken (i) as an abstract version of the \textit{defeasibility theory of knowledge} (cf. \cite{Lehrer1969}, \cite{Baltag2022}) and (ii) as a novel implementation of the insight from \cite{Hacking1967} that the \textit{available information} bears on whether a speaker is entitled to an epistemic possibility claim, going beyond the actual knowledge of the speaker or hearers.

\section{\textit{Domain Semantics}}\label{section-rivals}

Domain semantics invites a natural account of knowledge ascription that exhibits \textbf{NTrans}. This contrasts with the influential \textit{descriptivist/factualist} school on epistemic modals, according to which `it might be that $p$' is taken as synonymous with, roughly, `$p$ is not ruled out by what is mutually known, or easily known, by a relevant group of agents'. Negative transparency seems untenable on the descriptivist account: that Smith knows that the train isn't late does not entail that Smith knows anything about what the mutual knowledge of a certain group rules out (even if the group includes only Smith: she might well be uncertain what she knows).  

An \textit{information model} $\mathcal{I} = \langle W, \mathtt{I} \rangle$ is a pair, with $W$ the set of all possible worlds and $\mathtt{I}$ an assignment of an information state $\mathtt{I}(p)$ to each atomic sentence of $\mathcal{L}$. We take an information state -- generically denoted $\textbf{i}$ -- to just be an \textit{intension}, i.e., a subset of $W$. State $\textbf{i}$ is \textit{veridical at $w$} when $w \in \textbf{i}$. We evaluate sentences in $\mathcal{L}$ as true (1) or false (0) relative to a possible world $w$ and an information state $\textbf{i}$: the valuation function $[ \cdot ]^{w, \textbf{i}}$ is as follows.

\begin{definition}[\textbf{Domain Semantics}] Given an information model $\mathcal{I}$:

\begin{tabular}{ l c l }
	$[ p ]^{w, \textbf{i}} = 1$ & iff & $w \in \mathtt{I}(p)$\\
	$[ \neg \varphi ]^{w, \textbf{i}} = 1$ & iff & $[ \varphi ]^{w, \textbf{i}} = 0$\\
	$[ \varphi \wedge \psi ]^{w, \textbf{i}} = 1$ & iff & $[ \varphi ]^{w, \textbf{i}} = 1$ and $[ \psi ]^{w, \textbf{i}} = 1$\\
	$[ \diamond \varphi ]^{w, \textbf{i}} = 1$ & iff & $\exists u \in \textbf{i}$: $[ \varphi ]^{u, \textbf{i}} = 1$\\
\end{tabular}
\end{definition}

The following notion (following \cite{Yalcin2007}) will be important for our account of attitude ascriptions:

\begin{definition}[\textbf{Acceptance}]
$\textbf{i} \Vdash \varphi$ iff $\forall w \in \textbf{i}$: $[ \varphi ]^{w, \textbf{i}} = 1$
\end{definition}

If $\textbf{i} \Vdash \varphi$, we say information $\textbf{i}$ \textit{accepts} or \textit{supports} sentence $\varphi$, modeling the idea that having exactly the information $\textbf{i}$ is sufficient for establishing $\varphi$, rendering $\varphi$ correctly assertable (putting aside Gricean considerations, anyway). To get a feel for $\Vdash$, note that the following sensible properties are readily verified (though note that, given domain semantics, they do not generalize; cf. \S \ref{section-bluntness-redux}, \cite{Hawke2021}):

\medskip

\begin{tabular}{ l c l } 
$\textbf{i} \Vdash p$ & iff & $\forall w \in \textbf{i}$: $w \in \mathtt{I}(p)$\\
$\textbf{i} \Vdash \neg p$ & iff & $\forall w \in \textbf{i}$: $w \notin \mathtt{I}(p)$\\
$\textbf{i} \Vdash p \wedge q$ & iff & $\textbf{i} \Vdash p$ and $\textbf{i} \Vdash q$\\
$\textbf{i} \Vdash p \vee q$ & iff & $\exists \textbf{i}_1, \textbf{i}_2$ s.t. $\textbf{i} = \textbf{i}_1 \cup \textbf{i}_2$ and $\textbf{i}_1 \Vdash p$ and $\textbf{i}_2 \Vdash q$\\
$\textbf{i} \Vdash \diamond p$ & iff & $\exists w \in \textbf{i}$: $\{w\} \Vdash p$\\
$\textbf{i} \Vdash \neg \diamond p$ & iff & $\forall w \in \textbf{i}$: $\{w\} \Vdash \neg p$\\
\end{tabular}

\medskip

As for logical consequence, two notions of entailment are prominent in this framework. First, a truth-preservation relation $\vDash$ is straightforwardly defined: $\varphi \vDash \psi$ holds exactly when $[ \varphi ]^{w, \textbf{i}} = 1$ implies $[ \psi ]^{w, \textbf{i}} = 1$ for every $w$ and $\textbf{i}$ in every model $\mathcal{I}$. Second, an acceptance-preservation relation $\Vdash$ is straightforwardly defined: $\varphi \Vdash \psi$ holds exactly when $\textbf{i} \Vdash \varphi$ implies $\textbf{i} \Vdash \psi$ for every $\textbf{i}$ in every model $\mathcal{I}$. Both consequence relations serve as useful tools for explaining ordinary intuitions about entailment and contradiction. For example, the domain semanticist utilizes $\Vdash$, not $\vDash$, to explain the incoherence of epistemic contradictions of the form $p \wedge \diamond \neg p$: while $p \wedge \diamond \neg p$ is consistent with respect to $\vDash$, there is no $\textbf{i}$ such that $\textbf{i} \Vdash p \wedge \lozenge \neg p$.

To introduce attitude ascriptions, we transfer an account of belief ascription from \cite{Yalcin2011} to knowledge ascription. Call this the \textit{classical approach}. A \textit{classical model} $\mathcal{C}$ supplements an information model with function $\textbf{k}$, mapping a world to a non-empty intension $\textbf{k}^w$. The idea is that $\textbf{k}^w$ models Smith's epistemic state at $w$ as a set of \textit{epistemic alternatives} (the total informational content of Smith's knowledge). As an agent's knowledge can never rule out the actual world, we stipulate:

\begin{itemize}
\item[C1.] $\forall w \in W$: $w \in \textbf{k}^w$
\end{itemize}

\begin{definition}[\textbf{Classicism}] Given classical $\mathcal{C}$, we extend domain semantics with:

\begin{tabular}{ l c l }
	$[ K\varphi ]^{w, \textbf{i}} = 1$ & iff & $\textbf{k}^w \Vdash \varphi$\\
\end{tabular}
\end{definition}

However, relative to the strategy of \S \ref{section-strategy}, classicism is only a partial success.

\begin{fact}
For classicists, \textbf{NTrans} holds.
\end{fact}

\begin{proof} $[ K\neg \diamond p ]^{w, \textbf{i}}$ iff $\textbf{k}^w \Vdash \neg \diamond p$ iff $\forall u \in \textbf{k}^w$: $\{u\} \Vdash \neg p$ iff $\forall u \in \textbf{k}^w$: $u \notin \mathtt{I}(p)$ iff $\textbf{k}^w \Vdash \neg p$ iff $[ K \neg p ]^{w, \textbf{i}}$ \end{proof}

\begin{fact}
For classicists, \textbf{Ver} fails.
\end{fact}

\begin{proof} Counter-model: consider $\mathcal{C}$ where (i) $W = \{w_1, w_2\}$, (ii) $\mathtt{I}(p) = \{w_2\}$, (iii) $\textbf{k}^{w_1} = W$. Let $\textbf{i} = \{w_1\}$. So, by (ii) and (iii), $[ K\diamond p ]^{w_1, \textbf{i}} = 1$, as there is a $p$-world in $\textbf{k}^{w_1}$. But $[ \diamond p ]^{w_1, \textbf{i}}=0$, as there is no $p$-world in $\textbf{i}$.\end{proof}

Of course, a small modification to the semantics secures \textbf{Ver}: 

\medskip

\begin{tabular}{ l c l }
$[ K\varphi ]^{w, \textbf{i}} = 1$ & iff: & $\textbf{k}^w \Vdash \varphi$ and $[ \varphi ]^{w, \textbf{i}} = 1$.
\end{tabular} 

\medskip

However, the modified proposal abandons \textbf{NTrans}. For a counter-model, take $\mathcal{C}$ where, for some $@ \in W$, every world in $\textbf{k}^@$ (including $@$ itself) is a $\neg p$-world (assuring $\textbf{k}^@ \Vdash \neg p \wedge \neg \diamond p$ and $[ \neg p ]^{@, \textbf{i}} = 1$), but there is a $p$-world in $\textbf{i}$ (so $[ \neg \diamond p ]^{@, \textbf{i}} = 0$). So, given $\mathcal{C}$, $[ K\neg p ]^{@, \textbf{i}} = 1$ and $[ K\neg \diamond p ]^{@, \textbf{i}} = 0$.

However, it is readily checked that the modified proposal yields: $\textbf{i} \Vdash K\neg \diamond p$ iff $\textbf{i} \Vdash K\neg p$. So, \textbf{NTrans} emerges at the level of acceptance, in tandem with \textbf{Ver}. Nevertheless, two problems remain. First, the modified proposal is, as it stands, markedly \textit{ad hoc}: adding the clause $[ \varphi ]^{w, \textbf{i}} = 1$ to the truth condition for $K\varphi$ raises interpretive questions about the nature of $\textbf{k}^w$ and serves \textit{purely} to assure factivity in the case of modalized formulas (it is readily checked that \textbf{Ver} holds for $\lozenge$-free formulas in the original account of $K\varphi$). Second, even more pointedly, the modified proposal does not yield \textbf{GeNT}: in particular, there exists $\mathcal{C}$ and $\textbf{i}$ where $\textbf{i} \Vdash K\neg (p \wedge q)$ but $\textbf{i} \nVdash K\neg (p \wedge \diamond q)$. To see this, let $\textbf{i}$ contain only worlds $w_1$ and $w_2$, with $p$ only true at $w_1$, and $q$ only true at $w_2$. Thus, $\textbf{i} \Vdash \neg (p \wedge q)$ but $\textbf{i} \nVdash \neg (p \wedge \diamond q)$ (as $[p \wedge \diamond q]^{w_1, \textbf{i}} = 1$). If we further set $\textbf{k}^w$ to be $\textbf{i}$ for every $w \in \textbf{i}$, we get: $\textbf{i} \Vdash K \neg (p \wedge q)$ but $\textbf{i} \nVdash K \neg (p \wedge \diamond q)$.

\section{\textit{Stable Acceptance Semantics}}\label{section-bluntness-redux}

We now present an information-sensitive semantic theory that achieves the goal of \S \ref{section-strategy}. The leading idea behind this theory is that Smith's knowledge at $w$ is \textit{stable} under refinement of her veridical information at $w$ - or at least refinements that are `available' at $w$, in a sense to be clarified. 

Our system may be seen as a novel implementation of a well-known (alleged) insight that the truth/aptness of an epistemic possibility claim is sensitive to \textit{objective factors} that go beyond the actual knowledge of the speaker or other relevant agents: in particular, it is sensitive to information that has not been acquired but is (in some sense) \textit{available} to the relevant agents. Consider two cases from \cite{Hacking1967}.

\begin{quote}
Imagine a salvage crew searching for a ship that sank a long time ago. The mate of the salvage ship works from an old log, makes a mistake in his calculations, and concludes that the wreck may be in a certain bay. It is possible, he says, that the hulk is in these waters. No one knows anything to the contrary. But in fact, as it turns out later, it simply was not possible for the vessel to be in that bay; more careful examination of the log shows that the boat must have gone down at least thirty miles further south. The mate said something false when he said, ``It is possible that we shall find the treasure here'', but the falsehood did not arise from what anyone actually knew at the time. \cite[pg. 148]{Hacking1967}
\end{quote}

As for the second case:

\begin{quote}
Consider a person who buys a lottery ticket. At the time he buys his ticket we shall say it is possible he will win, though probably he will not. As expected, he loses. But retrospectively it would be absurd to report that it only seemed possible that the man would win. It was perfectly possible that he would win. To see this clearly, consider a slightly different case, in which the lottery is not above board; it is rigged so that only the proprietors can win. Thus, however it may have seemed to the gullible customer, it really was not possible that he would win. It only seemed so. ``Seemed possible'' and ``was possible'' both have work cut out for them. \cite[pg. 148]{Hacking1967}
\end{quote}

This suggests a proposal along the following lines: that whether an epistemic possibility claim is aptly assertible depends, in context, not only on the information that is already possessed, but that is available via ``practicable investigation'' (as Hacking puts it), or depends (as \cite{DeRose1991} puts it) on the ``relevant way[s] by which members of the relevant community can come to know'', or tracks (as \cite[pg. 402]{Moore1962} puts it) a distinction between what the speaker or other relevant agents ``easily might know'' versus ``couldn't easily know or have known''. We needn't commit to any particular elaboration here (cf.\cite[sect.1]{Fintel2011}).

Exactly what to make of the above cases is debatable, as \cite[Sects. 10.2.2, 10.4.2]{MacFarlane2014} points out. For our purposes, we need only observe the following. First, one hesitates to say that the mate \textit{knew} that they might find the treasure in the bay: as his claim could not be maintained were accessible further evidence collected, it does not rise to knowledge. Second, it seems reasonable to say that we \textit{knew}, at the time, that the person with the fair lottery ticket might win (but probably would not). Our beliefs seemed sufficiently sensitive to the available information: given the intrinsic limits on predicting a lottery, the possibility of his winning could not be ruled out even with all accessible evidence on the table.

Two strategies are available to theorists for explaining these observations. First, one could incorporate objective factors as a constraint on \textit{epistemic possibility claims}. As \cite[Sect. 10.2.2]{MacFarlane2014} notes, this has the cost that it becomes hard to see how the casual `might' claims we make in ordinary life are ever warranted. Alternatively, one could incorporate objective factors as a constraint on \textit{knowledge ascriptions} (with an eye to delivering plausible interactions with epistemic modals). As the conditions for asserting a knowledge claim are plausibly relatively demanding, the analogue of the previous objection has less force in this case. Our own theory exploits this second approach, citing the precedent and independent motivation provided by the tradition of \textit{defeasibility} theories of knowledge, in the spirit of \cite{Lehrer1969} (we leave more detailed comparisons for elsewhere).

In contrast to domain semantics, we offer a bilateral \textit{acceptance semantics}: instead of evaluating sentences at world-information pairs and deriving acceptance conditions, sentences are evaluated at just an information state. Hence, acceptance conditions (and, simultaneously, rejection conditions) are \textit{directly} provided. For some independent advantages of working with an acceptance semantics, see \cite{Veltman1985}, \cite{Schroeder2008a}, \cite{Ciardelli2021} and \cite{Aloni2022}; for independent drawbacks to domain semantics, see \cite{Hawke2021}.

A \textit{bounded model} $\mathcal{M}$ supplements an information model with functions $\textbf{k}$ and $\textbf{i}$, each mapping a world to an information state (a non-empty intension), respectively denoted $\textbf{k}^w$ and $\textbf{i}^w$. We call $\textbf{i}^w$ the \textit{worldly information at $w$}, while $\textbf{k}^w$ again models the set of \textit{epistemic alternatives}: the possible worlds compatible with the agent's total knowledge state (for simplicity we proceed with a single agent, writing $K$ instead of $K_1$). We say that intension $\textbf{j}$ \textit{refines} intension $\textbf{i}$ when $\textbf{j} \subseteq \textbf{i}$. We say that $\textbf{i}$ is \emph{internally coherent} when $\textbf{i}$ is non-empty and, for every $w \in \textbf{i}$, $\textbf{i}^w$ refines $\textbf{i}$. Intuitively, an internally coherent information state $\textbf{i}$ is coherent in the following sense: if $\textbf{i}$ leaves it open that the best available information (the `worldly information') cannot rule out a certain possibility, then $\textbf{i}$ does not itself rule out that possibility. We say that $\textbf{i}$ is \textit{accessible} at $w$ exactly when $\textbf{i}$ is both internally coherent and veridical at $w$, i.e., $w \in \textbf{i}$. We stipulate, for all $w \in W$, that $\textbf{k}^w$ and $\textbf{i}^w$ are both accessible at $w$.

\begin{lemma}\label{lemma-SAS-1}
If $\textbf{i}$ is internally coherent then $\textbf{i} = \bigcup_{w \in \textbf{i}} \textbf{i}^w$.
\end{lemma}

\begin{proof} As $\textbf{i}^w$ refines $\textbf{i}$ for all $w \in \textbf{i}$, we have $\bigcup_{w \in \textbf{i}} \textbf{i}^w \subseteq \textbf{i}$. Suppose that $w \in \textbf{i}$. As $\textbf{i}^w$ is accessible at $w$, $w \in \textbf{i}^w$. So, $\textbf{i} \subseteq \bigcup_{w \in \textbf{i}} \textbf{i}^w $. \end{proof}

\begin{definition}[\textbf{Accessible Refinement}]
Given information state $\textbf{i}$, let $Acc(\textbf{i})$ be the set of information states $\textbf{j}$ where (i) $\textbf{j}$ refines $\textbf{i}$ and (ii) $\textbf{j}$ is accessible at $w$ for some $w \in \textbf{i}$. We call the members of $Acc(\textbf{i})$ the \textit{accessible refinements of $\textbf{i}$}. \end{definition}

Note that every $\textbf{j} \in Acc(\textbf{i})$ has the property: there exists $w \in \textbf{i}$ such that $\textbf{i}^w \subseteq  \textbf{j} \subseteq \textbf{i}$. Thus, the accessible refinements of $\textbf{i}$ are bounded by the candidates left open by $\textbf{i}$ for what the worldly information might be.

\begin{definition}[\textbf{Stable Acceptance Semantics}]
Given bounded $\mathcal{M}$, intension $\textbf{i}$:

\begin{tabular}{l c l}
$\textbf{i} \Vdash p$ & iff & $\forall w \in \textbf{i}$: $w \in \mathtt{I}(p)$\\
$\textbf{i} \dashV p$ & iff & $\forall w \in \textbf{i}$: $w \notin \mathtt{I}(p)$\\
$\textbf{i} \Vdash \neg \varphi$ & iff & $\textbf{i} \dashV \varphi$\\
$\textbf{i} \dashV \neg \varphi$ & iff & $\textbf{i} \Vdash \varphi$\\
$\textbf{i} \Vdash \varphi \wedge \psi$ & iff & $\textbf{i} \Vdash \varphi$ and $\textbf{i} \Vdash \psi$\\
$\textbf{i} \dashV \varphi \wedge \psi$ & iff & $\exists \textbf{i}_1, \textbf{i}_2$ s.t. $\textbf{i} = \textbf{i}_1 \cup \textbf{i}_2$ and $\textbf{i}_1 \dashV \varphi$ and $\textbf{i}_2 \dashV \psi$\\
$\textbf{i} \Vdash \diamond \varphi$ & iff & $\exists w \in \textbf{i}$: $\{w\} \Vdash \varphi$\\
$\textbf{i} \dashV \diamond \varphi$ & iff & $\forall w \in \textbf{i}$: $\{w\} \dashV \varphi$\\
$\textbf{i} \Vdash K\varphi$ & iff & $\forall w \in \textbf{i}$, $\forall \textbf{j} \in Acc(\textbf{k}^w)$: $\textbf{\textbf{j}} \Vdash \varphi$\\
$\textbf{i} \dashV K\varphi$ & iff & $\forall w \in \textbf{i}$, $\exists \textbf{j} \in Acc(\textbf{k}^w)$: $\textbf{j} \nVdash \varphi$\\
\end{tabular}
\end{definition}

\medskip

Read $\textbf{i} \Vdash \varphi$ as `$\textbf{i}$ accepts $\varphi$' or `$\textbf{i}$ supports $\varphi$', and $\textbf{i} \dashV \varphi$ as `$\textbf{i}$ rejects $\varphi$' or `$\textbf{i}$ refutes $\varphi$'. The most unusual entry (cf. \cite{Veltman1985}, \cite{Hawke2021}, \cite{Aloni2022} and \S \ref{section-rivals}) is that for $K\varphi$: according to our semantics, `Smith knows that $\varphi$' can be accepted exactly when it is established that every accessible refinement of Smith's knowledge state supports $\varphi$; `Smith knows that $\varphi$' can be rejected exactly when it is established that an accessible refinement of Smith's knowledge state doesn't support $\varphi$.

A couple of technical lemmas will prove useful.

\begin{lemma}\label{tech-lemma1}
If $\textbf{i} \Vdash \varphi$ and $\textbf{j} \Vdash \varphi$ then $\textbf{i} \cup \textbf{j} \Vdash \varphi$. Likewise, if $\textbf{i} \dashV \varphi$ and $\textbf{j} \dashV \varphi$ then $\textbf{i} \cup \textbf{j} \dashV \varphi$.
\end{lemma}

\begin{proof} This can be established by a routine induction on $\varphi$, with respect to the following stronger property: (i) if $\textbf{a} \Vdash \varphi$ and $\textbf{b} \Vdash \varphi$ for all $\textbf{a} \subseteq \textbf{i}$ and $\textbf{b} \subseteq \textbf{j}$ then $\textbf{a} \cup \textbf{b} \Vdash \varphi$ for all $\textbf{a} \subseteq \textbf{i}$ and all $\textbf{b} \subseteq \textbf{j}$ and (ii) if $\textbf{a} \dashV \varphi$ and $\textbf{b} \dashV \varphi$ for all $\textbf{a} \subseteq \textbf{i}$ and all $\textbf{b} \subseteq \textbf{j}$ then $\textbf{a} \cup \textbf{b} \dashV \varphi$ for all $\textbf{a} \subseteq \textbf{i}$ and all $\textbf{b} \subseteq \textbf{j}$.  \end{proof}

\begin{lemma}\label{tech-lemma2}
If $\textbf{i}$ is internally coherent, the following are equivalent: 
\begin{itemize} 
\item[A.] $\forall \textbf{j} \in Acc(\textbf{i})$: $\textbf{\textbf{j}} \Vdash \varphi$
\item[B.] $\forall u \in \textbf{i}$: $\textbf{\textbf{i}}^u \Vdash \varphi$
\end{itemize}
\end{lemma}

\begin{proof} As $\textbf{i}^u \in Acc(\textbf{i})$ for all $u \in \textbf{i}$, the direction from A to B is trivial. For the other direction, consider $\textbf{j} \in Acc(\textbf{i})$ and use a routine induction on the structure of $\varphi$ to show that if $\textbf{\textbf{i}}^u \Vdash \varphi$ holds for all $u \in \textbf{i}$, then $\textbf{\textbf{j}} \Vdash \varphi$ holds, and if $\textbf{\textbf{i}}^u \dashV \varphi$ holds for all $u \in \textbf{i}$, then $\textbf{\textbf{j}} \dashV \varphi$ holds, with Lemma \ref{lemma-SAS-1} and Lemma \ref{tech-lemma1} being put to crucial use (the latter for the case of $\varphi \wedge \psi$). \end{proof}

It follows that our entries for $K\varphi$ have the following convenient reformulation, which we deploy in coming proofs:

\medskip

\begin{tabular}{l c l}
$\textbf{i} \Vdash K\varphi$ & iff & $\forall w \in \textbf{i}$, $\forall u \in \textbf{k}^w$: $\textbf{\textbf{i}}^u \Vdash \varphi$\\
$\textbf{i} \dashV K\varphi$ & iff & $\forall w \in \textbf{i}$, $\exists u \in \textbf{k}^w$: $\textbf{i}^u \nVdash \varphi$\\
\end{tabular}

\medskip

Thus, according to our semantics, `Smith knows that $\varphi$' can be accepted exactly when it is established that Smith's knowledge state establishes that the worldly information establishes $\varphi$; `Smith knows that $\varphi$' can be rejected exactly when it is established that Smith's knowledge state leaves it open that the worldly information doesn't establish $\varphi$.

This system invites the following notion of logical consequence:

 \begin{definition}[\textbf{Coherent Consequence}]
 $\varphi \VDash \psi$ iff, for every bounded model $\mathcal{M}$, if $\textbf{i}$ is internally coherent and $\textbf{i} \Vdash \varphi$, then $\textbf{i} \Vdash \psi$.
 \end{definition}

\begin{definition}[\textbf{Assertoric Equivalence}]
Sentences $\varphi$ and $\psi$ are \emph{assertorically equivalent} if
$$ \textbf{i} \Vdash \varphi \textrm{ iff } \textbf{i} \Vdash \psi $$
for every information state $\textbf{i}$ in every bounded model $\mathcal{M}$.
\end{definition}

For example, $p \wedge q$ and $q \wedge p$ are assertorically equivalent.

\begin{definition}\label{definition-might-restricted}
A sentence $\varphi$ is \emph{$\diamond$-restricted} if the only occurrences of $\diamond$ are in the scope of a $K$ operator.
\end{definition}

For example, $\neg (p \wedge q)$ and $K\diamond p$ are $\diamond$-restricted; $\diamond p$ and $\neg \diamond (p \vee q)$ aren't.

To efficiently demonstrate the key properties of our system, we require some preliminary results, which are of independent technical interest.

\begin{lemma}\label{lemma-SAS2}
If $\varphi$ is $\diamond$-restricted then:
\begin{enumerate}
\item $\textbf{i} \Vdash \varphi$ iff $\forall w \in \textbf{i}$: $\{w\} \Vdash \varphi$
\item $\textbf{i} \dashV \varphi$ iff $\forall w \in \textbf{i}$: $\{w\} \dashV \varphi$
\end{enumerate}
\end{lemma}

\begin{proof} A routine induction. \end{proof}

\begin{lemma}\label{lemma-SAS4}
If $\varphi$ is $\diamond$-restricted then: $\textbf{i} \dashV \diamond \varphi \;  \textrm{ iff } \; \textbf{i} \dashV \varphi$
\end{lemma}

\begin{proof} Suppose that $\textbf{i} \dashV \diamond \varphi$. Thus, $\forall w \in \textbf{i}$: $\{w\} \dashV \varphi$. Thus, by Lemma \ref{lemma-SAS2}, $\textbf{i} \dashV \varphi$. The reasoning can be reversed.\end{proof}

\begin{theorem}[Normal Form]\label{normal-form}
For every sentence $\varphi$, there exists $n \geq 0$ and $\diamond$-restricted sentences $\alpha_0$, $\alpha_1$, \ldots, $\alpha_n$ such that for any internally coherent $\textbf{i}$:
$$ \textbf{i} \Vdash \varphi \textrm{ iff } \textbf{i} \Vdash \alpha_0 \wedge \diamond\alpha_1 \wedge \cdots \wedge \diamond \alpha_n $$
\end{theorem}

\begin{proof} See the appendix. \end{proof}

Now for the key results.

\begin{fact}\label{BB-NTrans}
\textbf{Generalized Negative Transparency} holds: $K\neg(p \wedge \diamond q) \DashV \VDash K \neg (p \wedge q)$.
\end{fact}

\begin{proof} Suppose that $\textbf{i} \Vdash K \neg (p \wedge \diamond q)$. So, $\forall w \in \textbf{i}$, $\forall u \in \textbf{k}^w$: $\textbf{1}^u \Vdash \neg p$ and $\textbf{2}^u \Vdash \neg \diamond q$, where $\textbf{1}^u \cup \textbf{2}^u = \textbf{i}^u$. By Lemma \ref{lemma-SAS4}: $\forall w \in \textbf{i}$, $\forall u \in \textbf{k}^w$: $\textbf{2}^u \Vdash \neg q$. So, $\textbf{i} \Vdash K \neg(p \wedge q)$. The reasoning can be reversed. \end{proof}

\begin{fact}\label{BB-factivity}
\textbf{K-Veridicality} holds: $K\varphi \VDash \varphi$.
\end{fact}
  
\begin{proof} Assume that $\textbf{i}$ is internally coherent and $\textbf{i} \Vdash K\varphi$. So, $\forall w \in \textbf{i}$, $\forall u \in \textbf{k}^w$: $\textbf{i}^u \Vdash \varphi$. By Theorem \ref{normal-form}, there exists $n \geq 1$ and $\diamond$-restricted sentences $\alpha_0, \alpha_1, \ldots, \alpha_n$ such that, $\forall w \in \textbf{i}$, $\forall u \in \textbf{k}^w$: $\textbf{i}^u \Vdash \alpha_0 \wedge \diamond\alpha_1 \wedge \cdots \wedge \diamond \alpha_n $.

We show that $\textbf{i} \Vdash \alpha_0$. Let $w \in \textbf{i}$. Now, as $w \in \textbf{k}^w$ and $\textbf{i}^u \Vdash \alpha_0$ for any $u \in \textbf{k}^w$, we have $\textbf{i}^w \Vdash \alpha_0$. So, by Lemma \ref{lemma-SAS2}, we have $\forall u \in \textbf{i}^w$: $\{u\} \Vdash \alpha_0$. Thus, as $w \in \textbf{i}^w$, we have $\{w\} \Vdash \alpha_0$. Generalizing: $\forall w \in \textbf{i}$: $\{w\} \Vdash \alpha_0$. So, by Lemma \ref{lemma-SAS2}, $\textbf{i} \Vdash \alpha_0$.

We show that $\textbf{i} \Vdash \diamond \alpha_k$ for $1 \leq k \leq n$. Let $w \in \textbf{i}$. Now, for any $u \in \textbf{k}^w$, there exists $v \in \textbf{i}^u$ such that $\{v\} \Vdash \alpha_k$, as $\textbf{i}^u \Vdash \diamond \alpha_k$. As $w \in \textbf{k}^w$, it follows that there exists $v \in \textbf{i}^w$ such that $\{v\} \Vdash \alpha_k$. Thus, as $\textbf{i}$ is internally coherent, $\exists v \in \textbf{i}$ such that $\{v\} \Vdash \alpha_k$. So, $\textbf{i} \Vdash \diamond \alpha_k$.

Altogether: $\textbf{i} \Vdash \alpha_0 \wedge \diamond\alpha_1 \wedge \cdots \wedge \diamond \alpha_n $. So, by Theorem \ref{normal-form}, $\textbf{i} \Vdash \varphi$. \end{proof}

It is instructive to linger on the broad explanation as to why $\diamond p$ is a coherent consequence of $K \diamond p$. Suppose that $\textbf{i}$ is internally coherent and supports $K \diamond p$. Thus, $\textbf{i}$ establishes that Smith's knowledge state establishes that the worldly information establishes $\diamond p$. Thus, the candidates for the worldly information -- those $\textbf{i}$ cannot rule out -- all contain a $p$-world. As $\textbf{i}$ is internally coherent, $\textbf{i}$ cannot itself rule out these worlds. So, $\textbf{i}$ accepts $\diamond p$.

Finally:
 
\begin{fact}\label{BB-speaker-uniformity}
\textbf{Uniformity} fails: $\diamond p \nVDash \neg K \neg p$.
\end{fact}

\begin{proof} Consider any bounded model $\mathcal{M}$ where: (i) $w_1 \in \mathtt{I}(p)$ and $w_2 \notin \mathtt{I}(p)$; (ii) $\textbf{i}^{w_1} = \textbf{k}^{w_1} = \{w_1\}$ and $\textbf{i}^{w_2} = \textbf{k}^{w_2} = \{w_2\}$. Set $\textbf{i} = \{w_1, w_2\}$. Note that $\textbf{i}$ is internally coherent.

By (i), $\{w_1\} \Vdash p$. So, $\exists w \in \textbf{i}$: $\{w\} \Vdash p$. So, $\textbf{i} \Vdash \diamond p$.

By (i), $\{w_2\} \dashV p$. Thus, by Lemma \ref{lemma-SAS2} and (ii), $\textbf{i}^{w_2} \dashV p$. Thus, by (ii), $\forall u \in \textbf{k}^{w_2}$: $\textbf{i}^u \dashV p$. Thus, $\textbf{i} \ndashV K \neg p$. Thus, $\exists w \in \textbf{i}$ such that $\forall u \in \textbf{k}^w$: $\textbf{i}^u \dashV p$. Thus, $\textbf{i} \ndashV K \neg p$. Thus, $\textbf{i} \nVdash \neg K \neg p$.\end{proof}

\nocite{*}
\bibliographystyle{eptcs}
\bibliography{TARK}

\appendix

\section{Appendix: Normal Form for Acceptance Semantics}\label{section-normal-form}

\textbf{Theorem \ref{normal-form}}. For every sentence $\varphi$, there exists $n \geq 0$ and $\diamond$-restricted sentences $\alpha_0, \alpha_1, \ldots, \alpha_n$ such that for any internally coherent $\textbf{i}$:
$$ \textbf{i} \Vdash \varphi \textrm{ iff } \textbf{i} \Vdash \alpha_0 \wedge \diamond\alpha_1 \wedge \cdots \wedge \diamond \alpha_n $$

\begin{proof}
We proceed by induction on sentence structure, with respect to the following stronger property: there exists $m, n \geq 0$ and $\diamond$-restricted sentences $\alpha_0, \alpha_1, \ldots, \alpha_m$ and $\beta_0, \beta_1, \ldots, \beta_n$ such that, for any internally coherent $\textbf{i}$:
\begin{itemize}
\item[] $\textbf{i} \Vdash \varphi$ iff $\textbf{i} \Vdash \alpha_0 \wedge \diamond \alpha_1 \wedge \ldots \wedge \diamond \alpha_m$
\item[] $\textbf{i} \dashV \varphi$ iff $\textbf{i} \Vdash \beta_0 \wedge \diamond \beta_1 \wedge \ldots \wedge \diamond \beta_n$
\end{itemize}

The case for atom $p$ is trivial, as this sentence is itself $\diamond$-restricted: set $m=n=0$, $\alpha_0 = p$ and $\beta_0 = \neg p$.

The case for knowledge ascription $K\varphi$ is trivial, as this sentence is itself $\diamond$-restricted: set $m=n=0$, $\alpha_0 = K\varphi$ and $\beta_0 = \neg K \varphi$.

For the induction hypothesis IH, assume, for arbitrary $\varphi$ and $\psi$, that there exists $m, n, x, y \geq 0$ and $\diamond$-restricted sentences $$\alpha_0, \alpha_1, \ldots, \alpha_m, \beta_0, \beta_1, \ldots, \beta_n, \delta_0, \delta_1, \ldots, \delta_x, \epsilon_0, \epsilon_1, \ldots, \epsilon_y$$ such that, for any internally coherent $\textbf{i}$:

\begin{itemize}
\item[] $\textbf{i} \Vdash \varphi$ iff $\textbf{i} \Vdash \alpha_0 \wedge \diamond \alpha_1 \wedge \ldots \wedge \diamond \alpha_m$
\item[] $\textbf{i} \dashV \varphi$ iff $\textbf{i} \Vdash \beta_0 \wedge \diamond \beta_1 \wedge \ldots \wedge \diamond \beta_n$
\item[] $\textbf{i} \Vdash \psi$ iff $\textbf{i} \Vdash \delta_0 \wedge \diamond \delta_1 \wedge \ldots \wedge \diamond \delta_x$
\item[] $\textbf{i} \dashV \psi$ iff $\textbf{i} \Vdash \epsilon_0 \wedge \diamond \epsilon_1 \wedge \ldots \wedge \diamond \epsilon_y$
\end{itemize}

Using the IH, we can prove the following.

\begin{tabular}{l c l}
$\textbf{i} \Vdash \neg \varphi$ & iff &  $\textbf{i} \dashV \varphi$\\
& iff & $\textbf{i} \Vdash \beta_0 \wedge \diamond \beta_1 \wedge \ldots \wedge \diamond \beta_n$\\
\end{tabular}

\begin{tabular}{l c l}
$\textbf{i} \dashV \neg \varphi$ & iff & $\textbf{i} \Vdash \varphi$\\
& iff & $\textbf{i} \Vdash \alpha_0 \wedge \diamond \alpha_1 \wedge \ldots \wedge \diamond \alpha_m$\\
\end{tabular}

\begin{tabular}{l c l}
$\textbf{i} \Vdash \varphi \wedge \psi$ & iff & $\textbf{i} \Vdash \varphi$ and $\textbf{i} \Vdash \psi$\\
& iff & $\textbf{i} \Vdash \alpha_0 \wedge \diamond \alpha_1 \wedge \ldots \wedge \diamond \alpha_m$ and $\textbf{i} \Vdash \delta_0 \wedge \diamond \delta_1 \wedge \ldots \wedge \diamond \delta_x$\\
& iff & $\textbf{i} \Vdash (\alpha_0 \wedge \delta_0) \wedge \diamond \alpha_1 \wedge \ldots \wedge \diamond \alpha_m \wedge \diamond \delta_1 \wedge \ldots \wedge \diamond \delta_x$ \\
\end{tabular}

\begin{tabular}{l c l}
$\textbf{i} \dashV \varphi \wedge \psi$ & iff & $\exists \textbf{i}_1, \textbf{i}_2$: $\textbf{i} = \textbf{i}_1 \cup \textbf{i}_2$ and $\textbf{i}_1 \dashV \varphi$ and $\textbf{i}_2 \dashV \psi$\\
& iff & $\exists \textbf{i}_1, \textbf{i}_2$: $\textbf{i} = \textbf{i}_1 \cup \textbf{i}_2$ and $\textbf{i}_1 \Vdash \beta_0 \wedge \diamond \beta_1 \wedge \ldots \wedge \diamond \beta_n$\\
&& and $\textbf{i}_2 \Vdash \epsilon_0 \wedge \diamond \epsilon_1 \wedge \ldots \wedge \diamond \epsilon_y$\\
& iff & $\textbf{i} \Vdash (\beta_0 \vee \epsilon_0) \wedge \diamond (\beta_0 \wedge \beta_1) \wedge \ldots \wedge \diamond (\beta_0 \wedge \beta_m)$\\
&& \;\;\;\; $\wedge \diamond (\epsilon_0 \wedge \epsilon_1) \wedge \ldots \wedge \diamond (\epsilon_0 \wedge \epsilon_x)$\\
\end{tabular}

\begin{tabular}{l c l}
$\textbf{i} \Vdash \diamond \varphi$ & iff & $\exists w \in \textbf{i}$: $\{w\} \Vdash \varphi$\\
& iff & $\exists w \in \textbf{i}$: $\{w\} \Vdash \alpha_0 \wedge \diamond \alpha_1 \wedge \ldots \wedge \diamond \alpha_m$\\
& iff & $\exists w \in \textbf{i}$: $\{w\} \Vdash \alpha_0 \wedge \alpha_1 \wedge \ldots \wedge \alpha_m$\\
& iff & $\textbf{i} \Vdash \diamond (\alpha_0 \wedge \alpha_1 \wedge \ldots \wedge \alpha_m)$\\
& iff & $\textbf{i} \Vdash (p \vee \neg p) \wedge \diamond (\alpha_0 \wedge \alpha_1 \wedge \ldots \wedge \alpha_m)$\\
\end{tabular}

\begin{tabular}{l c l}
$\textbf{i} \dashV \diamond \varphi$ & iff & $\forall w \in \textbf{i}$: $\{w\} \dashV \varphi$\\
& iff & $\forall w \in \textbf{i}$: $\{w\} \Vdash \beta_0 \wedge \diamond \beta_1 \wedge \ldots \wedge \diamond \beta_n$\\
& iff & $\forall w \in \textbf{i}$: $\{w\} \Vdash \beta_0 \wedge \beta_1 \wedge \ldots \wedge \beta_n$\\
& iff & $\textbf{i} \Vdash \beta_0 \wedge \beta_1 \wedge \ldots \wedge \beta_n$\\
\end{tabular}

\end{proof}

\end{document}